\documentclass[conference,letterpaper]{IEEEtran}

\usepackage[utf8]{inputenc} 
\usepackage[T1]{fontenc}
\usepackage{url}
\usepackage{ifthen}
\usepackage{cite}
\usepackage{relsize}
\usepackage{algorithm}
\usepackage{algpseudocode}
\usepackage{amsfonts}
\usepackage{amsthm}
\usepackage{tikz}
\usepackage{pgfplots}
\pgfplotsset{compat=1.17}
\usepackage{xcolor}
\usepackage{balance} 
\usepackage[cmex10]{amsmath} 
\newtheorem{corollary}{Corollary}
\newtheorem{definition}{Definition}
\newtheorem{assumption}{Assumption}
\newtheorem{theorem}{Theorem}

\IEEEoverridecommandlockouts   

\begin{document}

\title{Differentially-Private Decentralized Learning in Heterogeneous Multicast Networks%
\thanks{This work is supported in part by the U.S. National Science Foundation under Grant 2107370.}%
} 


\author{%
  \IEEEauthorblockN{Amir Ziaeddini, Yauhen Yakimenka, Jörg Kliewer}
  \IEEEauthorblockA{Helen and John C. Hartmann Department of Electrical and Computer Engineering\\
                    New Jersey Institute of Technology, Newark, NJ, USA\\
                    Email: \{az328, yauhen.yakimenka, jkliewer\}@njit.edu}

}

\maketitle

\vspace{-5cm}

\begin{abstract}
   We propose a power-controlled differentially private decentralized learning algorithm designed for a set of clients aiming to collaboratively train a common learning model. The network is characterized by a row-stochastic adjacency matrix, which reflects different channel gains between the clients. In our privacy-preserving approach, both the transmit power for model updates and the level of injected Gaussian noise are jointly controlled to satisfy a given privacy and energy budget. We show that our proposed algorithm achieves a convergence rate of \(O(\log T)\), where \(T\) is the horizon bound in the regret function. Furthermore, our numerical results confirm that our proposed algorithm outperforms existing works.
\end{abstract}

\section{Introduction}

Recently, decentralized learning (DL) has emerged as a novel approach in the realm of distributed machine learning, addressing critical data analysis challenges posed by the explosion of data generated by mobile devices, laptops, internet-of-things (IoT) devices, and other communication systems worldwide (see, e.g., \cite{9220780,10542323,10044204 ,10680589}). Moreover, DL enables devices to train a shared model cooperatively through peer-to-peer communication. On the other hand, while DL inherently protects raw data by keeping it localized on client devices, it is not immune to privacy risks. Attack methods, such as inversion attacks\cite{10024757},\cite{10375767}, can exploit transmitted gradients or model coefficients to reconstruct sensitive information. Therefore, utilizing differential privacy (DP) as a robust mechanism to mitigate these risks by adding controlled noise to the shared information could preserve the privacy of clients' individual data by reducing information leakage (see, e.g., \cite{9714350,9069945,10705313}).

Most studies on differentially-private decentralized learning (DP-DL) assume that the network graph is either balanced or undirected, thereby implying that the weight matrix corresponding to the graph is doubly stochastic \cite{10025677, 10279097, 9517780, 9563232}. Despite ensuring privacy through DP, such methods ignore the inherent imbalances in communication networks, where different links exhibit distinct channel conditions. While the doubly-stochasticity assumption facilitates algorithm design and convergence analysis, it often lacks practicality in real-world communications scenarios. In \cite{6930814}, \cite{7405263}, and \cite{8268562}, the subgradient-push method was employed in order to relax the doubly stochastic assumption to a column-stochastic matrix, thereby eliminating the requirement for the graph to be balanced. In \cite{8988200} and \cite{8355917}, a push-sum-based algorithm was utilized, leveraging two weight matrices, one being row-stochastic and the other column-stochastic, to guarantee convergence over directed graphs. However, all of these approaches seem impractical as they impose significant constraints on graph weight selection and fail to effectively capture the dynamics of multicast-based communication. To address these problems, \cite{8267245} and \cite{7526803} utilize a row-stochastic weight matrix that allows all clients to independently determine the weight of the signal they want to multicast to their neighbors. However, these methods overlook privacy issues.

In \cite{9013030} and \cite{10552083} the authors employ Laplace noise to preserve DP under the assumption of a row-stochastic weight matrix. Also, \cite{10312089} introduces a locally-balanced noise-adding mechanism to address privacy. Although these studies have explored the integration of DP with DL in unbalanced directed graph configurations, they typically consider the weight matrix in a generic manner without incorporating a communication-centric perspective. Consequently, these approaches fail to account for critical factors such as node power levels and channel gains, and they lack mechanisms to proportionally balance each node's power budget between transmitting shared model coefficients and injecting noise, thereby hindering the achievement of optimal accuracy and privacy.

The novelty of this paper lies in extending a more practical scenario by considering heterogeneous channel gains, in contrast to \cite{10025677}, which assumes identical channel conditions, meaning that the channel gain between a node and all its neighbors is the same. To ensure DP, we employ the Gaussian mechanism and propose a novel power allocation strategy that enables all nodes to split their maximum transmission power between their primary signal and the noise required for privacy preservation. Also, we incorporate over-the-air computation (OAC) (see, e.g.,\cite{10025677,9563232,9322286}), which utilizes the property of wireless channels to aggregate signals directly during transmission, ensuring privacy below a predefined threshold, while being able to adapt to heterogeneous channel conditions.

\section{System Model}

We consider a wireless DL configuration with $K$ nodes, such as mobile devices or IoT sensors, each with its own local dataset $D_i$, collaboratively training a model while preserving data privacy. Each node communicates with its neighbors over wireless channels subject to varying channel conditions. To enhance communication efficiency, each node is equipped with two antennas: one dedicated to transmitting signals and the other to receiving signals enabling full-duplex communication. Furthermore, each node operates within a maximum power budget, which is divided between transmitting model coefficients and injecting noise to ensure DP. 

The network is represented as an unbalanced directed graph \( \mathcal{G} = (\mathcal{V}, \mathcal{E}) \), where \( \mathcal{V} = \{1, 2, \ldots, K\} \) is the set of nodes, and \( \mathcal{E} \subseteq \mathcal{V} \times \mathcal{V} \) is the set of directed edges. \( \mathcal{G} \) is strongly connected and fixed, i.e. there exists a directed path between any pair of nodes. A directed edge \( (i, j) \in \mathcal{E} \) indicates that node \( i \) can transmit information to node \( j \), making \( j \) a neighbor of \( i \). For simplicity, we assume that if \( (i, j) \in \mathcal{E} \), then \( (j, i) \in \mathcal{E} \). The neighbor set of node \( i \) is denoted by $\mathcal{N}_i = \{j|(i,j)\in \mathcal{E}\}$ and the number of neighbors of each node $i$ is \( d_i = |\mathcal{N}_i| \). Also, the maximum degree of the graph is defined by $R \stackrel{\triangle}{=} \max\limits_{i} d_i$. Each directed edge \( (i, j) \) corresponds to a wireless communication link, characterized by a channel gain that depends on factors such as distance, fading, and environmental interference. These directed and unbalanced edges reflect the network's heterogeneity, where some nodes may have significantly more or fewer neighbors than others.

In this setup, DL focuses on solving an optimization problem where \( K \) nodes collaboratively minimize a global objective function without relying on a central server. Each node \( i \) has access to its local cost function \( f_i : \mathbb{R}^m \to \mathbb{R} \), defined over a constraint set \( \Omega \subseteq \mathbb{R}^m \). The global objective function is
\begin{equation*}
F(\mathbf{x}) = \sum_{i=1}^K f_i(\mathbf{x}) = \sum_{i=1}^K f_i(\mathbf{x} ;D_{i}),   
\end{equation*}
where \( \mathbf{x} \in \Omega \). The goal is to minimize \( F(\mathbf{x}) \) by ensuring that the decisions \( \mathbf{x}_{i,t} \) made by node $i$ at time \( t \) converge to the optimal solution $\mathbf{x}^* = \arg\min_{\mathbf{x} \in \Omega} F(\mathbf{x})$.

To solve this problem in a decentralized setting, nodes communicate with their neighbors to exchange information about local updates. These interactions enable the network to collectively approximate the global solution \( \mathbf{x}^* \). We introduce the following commonly used assumptions\cite{9013030} which will be used throughout the remainder of this paper.


\begin{assumption}\label{a:s_convex_fi}
    The local cost functions $f_i$, $\forall i \in \mathcal{V}$, are $\mu$-strongly convex, meaning that for any $x$ and $y$, we have:
\[
f_i(\mathbf{y}) \geq f_i(\mathbf{x}) + \nabla f_i(\mathbf{x})^\top (\mathbf{y} - \mathbf{x}) + \frac{\mu}{2} \|\mathbf{y} - \mathbf{x}\|^2,
\]
where $\mu > 0$ is the strong convexity constant.
\end{assumption}

\begin{assumption}\label{a:omega}
    The constraint set \( \Omega \) is nonempty, convex, and closed. The diameter of \( \Omega \) is bounded by $L < \infty$.
\end{assumption}

To evaluate the performance, a regret function \cite{6930789} is defined for each node $i \in \mathcal{V}$ over a finite time $T$ as
\begin{equation*}
\begin{aligned}
\mathbb{R}_i(T) = \mathbb{E} \left[ \sum_{t=1}^T F(\mathbf{x}_{i,t}) \right] - \mathbb{E} \left[ \sum_{t=1}^T F(\mathbf{x}^*) \right].      
\end{aligned}
\end{equation*}

The regret function measures the difference between the cumulative expected cost incurred by the algorithm and that of the optimal solution over \( T \) epochs.

Incorporating DP into decentralized optimization focuses on safeguarding the privacy of each node's local data by adding noise to the model updates or coefficients exchanged between nodes, ensuring that sensitive information cannot be inferred from the shared data. To formally define the privacy guarantees offered by DP and its parameters, we provide the following definitions.

\begin{definition}[\!\cite{9714350}]
A mechanism \( M \) satisfies \((\epsilon, \delta)\)-DP if for any two neighboring datasets \( D \) and \( D' \) that differ in only one data point, and any set of outputs \( O \), the following holds:
\begin{equation*}
\mathbb{P}[M(D) \in O] \leq e^{\epsilon} \mathbb{P}[M(D') \in O] + \delta.
\end{equation*}    
\end{definition}

In this inequality, \(\epsilon\) quantifies the privacy budget, where smaller values indicate stronger privacy, and \(\delta\) accounts for a small probability of failure in providing privacy guarantees. Furthermore, in the context of DP we define the sensitivity by \(\Delta\) as the maximum change in the function’s output when a single data point in the input dataset is altered.

\begin{definition}[\!\cite{9714350}]
The sensitivity of a function $f$ is
\begin{equation*}
\Delta = \max_{D, D'} \|f(D) - f(D')\|,
\end{equation*}
where \(D\) and \(D'\) are neighboring datasets.
\end{definition}

Sensitivity plays a crucial role in determining the amount of noise that must be added to ensure the output adheres to the desired level of DP for any single data point.

The \textit{Gaussian} mechanism $F$ provides $(\epsilon, \delta)$-DP \cite{9714350} by adding Gaussian noise to the output of a function $f$ as
\begin{equation*}
    F(x) = f(x) + \mathcal{N}(0,\sigma^{2}),
\end{equation*}
where the standard deviation of the Gaussian injected noise is defined as
\begin{equation}\label{4}
    \sigma = \frac{\Delta}{\epsilon} \sqrt{2 \ln \frac{1.25}{\delta}}.
\end{equation}

\section{Proposed Algorithm}

In this section, we present our proposed scheme in two parts: (i) developing an algorithm that ensures convergence to a feasible solution and (ii) finding the required power fractions which tune the privacy and accuracy. 

In this model, $K$ nodes (clients) cooperatively aim to reach a common learning model by communicating with each other in a DP-DL framework. Each node $i \in \mathcal{V}$ only receives the signals that are synchronously transmitted from its neighbors. We denote by $\mathbf{x}_{i,t}$ the local parameters or coefficients of the $i$-th node at time $t$. During each epoch $t$, the transmitter $i$ injects some additive noise $\boldsymbol{\eta}_{i,t}$ to its local coefficients $\mathbf{x}_{i,t}$ to preserve DP, which protects sensitive 
data used to train the local models. The noise $\boldsymbol{\eta}_{i,t} \sim \mathcal{N}(0,\sigma_{i,t}^{2}) $ is drawn from a Gaussian distribution with zero mean and variance $\sigma_{i,t}^2$, ensuring that the injected noise meets some prescribed privacy budget requirements. Furthermore, since we want to balance the system performance and privacy leakage, each node $i$ divides its maximum transmit power $p_i$ into two components: (i) a fraction $\alpha_{i,t}$ is allocated for transmitting the local coefficients $\mathbf{x}_{i,t}$, (ii) the remaining fraction $\beta_{i,t} = 1 - \alpha_{i,t}$ is allocated for transmitting the injected noise $\boldsymbol{\eta}_{i,t}$.

Therefore, the total constructed signal in node $i \in \mathcal{V}$  for multicasting is as follows:
\begin{equation*}
    \tilde {\mathbf{x}}_{i,t}=\sqrt{\alpha_{i,t} p_{i}}\mathbf{x}_{i,t} + \sqrt{\beta_{i,t} p_{i}}\boldsymbol{\eta}_{i,t},
\end{equation*}
where $\mathbb{E}[\|\sqrt{\alpha_{i,t} p_{i}}\mathbf{x}_{i,t} + \sqrt{\beta_{i,t} p_{i}}\boldsymbol{\eta}_{i,t}\|^2]\leq p_i$. This multicasting scheme ensures that all neighboring nodes receive a combination of the local model coefficients and the privacy-preserving noise in a synchronous manner, thereby tuning the privacy of individual nodes' data while facilitating collaborative learning. Furthermore, the parameters $\alpha_i$ and $\beta_i$ are tuned to balance the trade-off between privacy leakage and system accuracy.

Each node $i \in \mathcal{V}$ multicasts its constructed signal to each of its neighbors $j \in \mathcal{N}_i$ through a unique channel determined by $|h_{ij}| e^{\zeta \phi_{ij}}$ where the channel characteristics are represented by two key components: channel gain $|h_{ij}|$ and channel phase $\phi_{ij}$. In this paper, we assume that the effect of the phase shift is not considered. In each epoch $t$, all neighbors of node $i$ transmit their privacy-preserving parameters across a shared communication channel, modeled as a Gaussian multiple access channel (MAC). Therefore, nodes transmit data directly through the wireless medium at the same time. The shared channel facilitates the aggregation of the transmitted privacy-preserving parameters. Thus, the received signal at each node $i \in \mathcal{V}$ is given as follows:
\begin{equation}\label{8}
\begin{aligned}
    \mathbf{y}_{i,t}&=\sum_{j \in \mathcal{N}_{i}} |h_{ji}|\tilde{\mathbf{x}}_{j,t}\\
    &= \sum_{j \in \mathcal{N}_{i}} |h_{ji}|(\sqrt{\alpha_{j,t} p_{j}}\mathbf{x}_{j,t} + \sqrt{\beta_{j,t} p_{j}}\boldsymbol{\eta}_{j,t}),
\end{aligned}    
\end{equation}
where channel noise is not considered for simplicity.

This communication procedure can be mapped to a graph network by the adjacency (weight) matrix  \( A = [a_{ij}] \in \mathbb{R}^{K \times K} \) of $\mathcal{G}$, where \( a_{ij} > 0 \) if \( (j, i) \in \mathcal{E} \), and \( a_{ij} = 0 \) otherwise. In our framework, with some modifications, this adjacency matrix is considered row-stochastic with positive self-loop weights, i.e. \( \sum_{j=1}^K a_{ij} = 1 \) , \( a_{ii} > 0 \), \( \forall i,j \in \mathcal{V} \). This structure supports unbalanced communication, where $h_{ij} \neq h_{ik}$ for $j,k \in \mathcal{N}_i$. It thus allows the DL process to efficiently integrate local interactions while ensuring that all nodes contribute to the global objective. Fig. 1 represents an example graph of DP-DL network with $K=4$ nodes in a diamond topology.  

\begin{figure}
    \centering
    \begin{tikzpicture}
    \hspace{0cm}
    \tikzset{
    thick node/.style={minimum size=1cm, inner sep=0, outer sep=0}
    }
    
    \node[thick node] (1) at (0, 0) {\includegraphics[width=1cm]{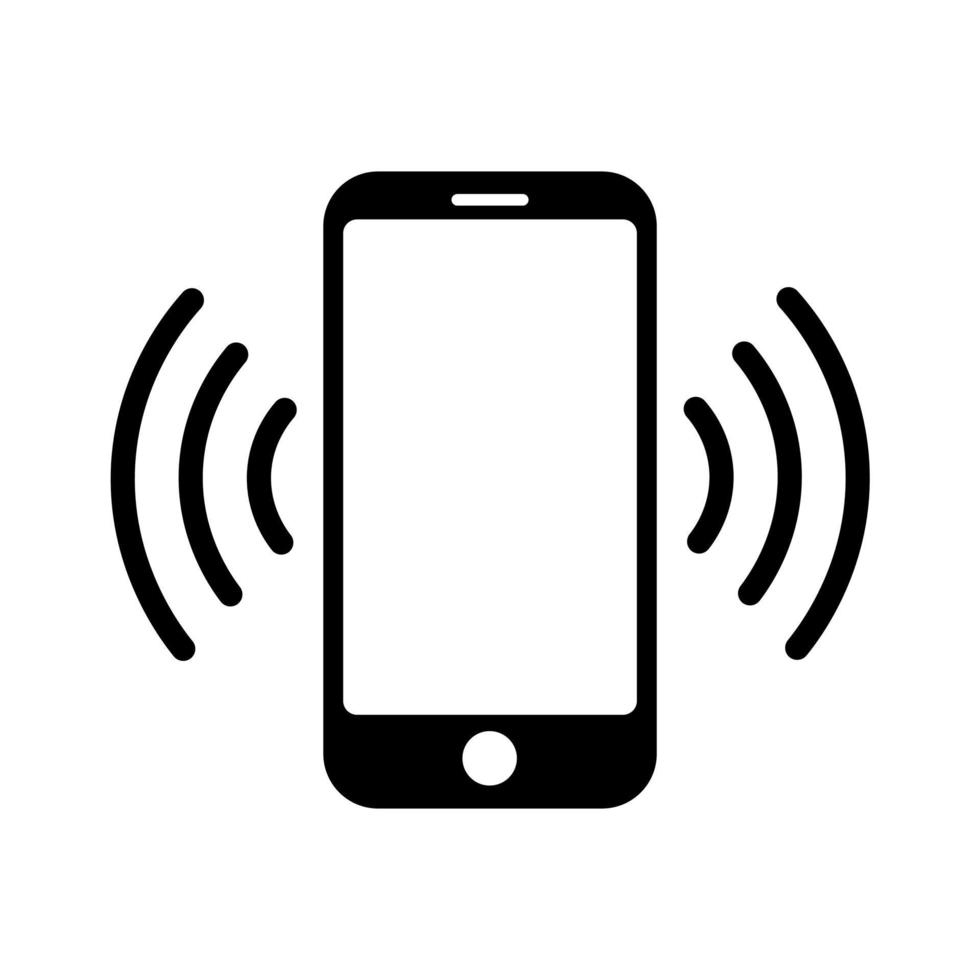}}; 
    \node[thick node] (2) at (1.8, 1.8) {\includegraphics[width=1cm]{Mobile_Phone.jpg}}; 
    \node[thick node] (3) at (3.6, 0) {\includegraphics[width=1cm]{Mobile_Phone.jpg}}; 
    \node[thick node] (4) at (1.8, -1.8) {\includegraphics[width=1cm]{Mobile_Phone.jpg}}; 

    \node at (1) {\textbf{1}};
    \node at (2) {\textbf{2}};
    \node at (3) {\textbf{3}};
    \node at (4) {\textbf{4}};
        
    \node[anchor=east] at ([xshift=-0.1cm] 1.west) {$
        \begin{aligned}
        &\sqrt{\alpha_{1,t} p_{1}} \mathbf{x}_{1,t} \\
        &\qquad +\\
        &\sqrt{\beta_{1,t} p_{1}} \boldsymbol{\eta}_{1,t}
        \end{aligned}
    $};
    \node[anchor=east] at ([xshift=2cm , yshift=0.3cm] 2.north) {$
        \begin{aligned}
        \sqrt{\alpha_{2,t} p_{2}} \mathbf{x}_{2,t} + \sqrt{\beta_{2,t} p_{2}} \boldsymbol{\eta}_{2,t}
        \end{aligned}
    $};

    \draw[->, ultra thick, blue] 
        (1) to[bend left] node[midway, above, xshift=-0.1cm] {$h_{12}$} (2);
    \draw[->, ultra thick, red] (2) to[bend left] node[midway, above, xshift=-0.1cm] {$h_{21}$} (1);

    \draw[->, ultra thick, red] (2) to[bend right] (3);
    \draw[->, ultra thick, green] (3) to[bend right] (2);

    \draw[->, ultra thick, green] (3) to[bend left] (4);
    \draw[->, ultra thick, yellow] (4) to[bend left] (3);

    \draw[->, ultra thick, yellow] (4) to[bend right] (1);
    \draw[->, ultra thick, blue] (1) to[bend right] (4);

    \end{tikzpicture}
        \caption{An illustrative diamond topology for the proposed DP-DL scheme with $K=4$ nodes, where each node splits its maximum power $p_i$ using scale factors $\alpha_{i,t}$ and $\beta_{i,t}$ for transmitting model coefficients $\mathbf{x}_{i,t}$ and injecting DP noise $\boldsymbol{\eta}_{i,t}$, sending signals over links with channel gains $h_{ij}$ at time $t$. }
        \label{fig:diamond-topology}
\end{figure}

In our problem, noise plays a critical role in achieving DP and impacts the convergence of the algorithm. However, to better understand how the algorithm behaves with DP noise, it is essential first to analyze its convergence properties under a noiseless scenario. Considering a noiseless scenario, the asymmetry of the row-stochastic adjacency matrix $A$ introduces challenges as the standard decentralized stochastic gradient descent (DSGD) algorithm cannot directly solve optimization problems over such graphs without adjustments\cite{8316938}. When $A$ is row-stochastic but not doubly-stochastic, its left Perron eigenvector \( \pi \) is non-uniform, leading to convergence toward the minimizer of a weighted global objective function $\bar{F}(\mathbf{x}) = \sum_{i=1}^K \pi_i f_i(\mathbf{x})$ instead of the original objective \( F(\mathbf{x}) = \sum_{i=1}^K f_i(\mathbf{x}) \). To address this issue, the left Perron eigenvector \( \pi \), which satisfies $\pi^\top A = \pi^\top$, is incorporated into the DSGD algorithm. By scaling the gradient term with \( \pi_i \), the update rule is modified as
\begin{equation*}
\mathbf{x}_{i,t+1} = \sum_{j=1}^K a_{ij} \mathbf{x}_{j,t} - \frac{\gamma_t}{\pi_i}  \mathbf{g}_{i,t} ,
\end{equation*}
where \( \gamma_t \) is the time-decreasing learning rate and 
$\mathbf{g}_{i,t} = \nabla f_{i}(\mathbf{x}_{i,t})$. 
This adjustment ensures that each node contributes proportionally to the global optimization, compensating for the unbalanced nature of the graph. The modified DSGD algorithm converges to the minimizer of the original objective \( F(\mathbf{x}) \), rather than the weighted objective \( \bar{F}(\mathbf{x}) \) \cite{8316938}.
Practically, an auxiliary variable $\mathbf{z}_{i,t}$ is designed to be exchanged between the clients in order to help in estimating $\pi_{i}$. 

In this decentralized setting, nodes iteratively update their decisions by combining local gradient information with messages received from their neighbors. The updates aim to ensure that all nodes converge to the same globally optimal solution \( \mathbf{x}^* \), despite having access only to their local cost functions.  Inspired by \cite{10025677} and \cite{9013030}, each node $i$ implements an update procedure in order to compute $\mathbf{x}_{i,t+1}$. During this step, each node updates its model coefficients based on the total received signal $\mathbf{y}_{i,t}$, its current model coefficients $\mathbf{x}_{i,t}$, the gradient of its loss function $\mathbf{g}_{i,t}$, auxiliary variable $\mathbf{z}_{i,t}$ and a compensation factor $c_{i,t}$ that is calculated to ensure convergence. The update rule is given by
\begin{equation}\label{1}
\begin{aligned}
    &\mathbf{x}_{i,t+1}
    = \mathlarger{\Pi_\Omega}\Bigg(\frac{\sum_{j \in \mathcal{N}_{i}} (|h_{ji}|\sqrt{\alpha_{j,t} p_{j}}\mathbf{x}_{j,t}+|h_{ji}|\sqrt{\beta_{j,t} p_{j}}\boldsymbol{\eta}_{j,t})}{c_{i,t} R} \\
    &+ \Big(1-\frac{d_{i}} {R}\Big) \Big(\mathbf{x}_{i,t} + \sqrt{\frac{\beta_{i,t}}{\alpha_{i,t}}} \boldsymbol{\eta}_{i,t}\Big)- \gamma_{t}\frac{\mathbf{g}_{i,t}}{z_{ii,t}}\Bigg),\\
\end{aligned}    
\end{equation}
where $\Pi_\Omega(.)$ is the projection function\cite{9013030} ensures that all the resulted $\mathbf{x}_{i,t+1}$ exist in the constraint set $\Omega$. This projection function $ \Pi_\Omega(\mathbf{x}) = \arg\min_{\mathbf{y} \in \Omega} \| \mathbf{x} - \mathbf{y} \| $ maps $ \mathbf{x} $ to the closest $ \mathbf{y} $ in the constraint set $ \Omega $.
 Furthermore, $z_{ii,t}$ is the $i$-th element of $\mathbf{z}_{i,t}$, which is an auxiliary vector transmitted along with $\mathbf{x}_{i,t}$. It has been shown in \cite{MAI201994} that the $z_{ii,t}$ converge to $\pi_i$ associated with the left Perron eigenvector of the adjacency matrix. Since the size of $\mathbf{z}_{i,t}$ is significantly smaller than the size of the model coefficients $\mathbf{x}_{i,t}$, the communication cost associated with transmitting $\mathbf{z}_{i,t}$ can be considered negligible. Furthermore, since $\mathbf{z}_{i,t}$ is independent of the data, it does not cause any information leakage.

Since we need to construct a row-stochastic adjacency matrix to guarantee the convergence of the update procedure of our algorithm, the compensation factor $c_{i,t}$ is given as $c_{i,t} = \frac{\sum_{j \in \mathcal{N}_{i}}|h_{ji}|(\sqrt{\alpha_{j,t} p_{j}})}{d_i}$.

After calculating all $c_{i,t}$ and considering all scale factors of $\mathbf{x}_{i,t}$ and $\boldsymbol{\eta}_{i,t}$, we can construct a closed form expression for the update procedure of $\mathbf{x}_{i,t}$ as
\begin{equation}\label{2}
\begin{aligned}
    \mathbf{x}_{i,t+1}= \mathlarger{\Pi_\Omega} \Bigg (\sum_{j=1}^{K} a_{ij} \Big (\mathbf{x}_{j,t} + \sqrt{\frac{\beta_{j,t}}{\alpha_{j,t}}} \boldsymbol{\eta}_{j,t}\Big) - \gamma_{t}\frac{\mathbf{g}_{i,t}}{z_{ii,t}}\Bigg),
\end{aligned}    
\end{equation}
where $a_{ij}=\frac{|h_{ji}| \sqrt{\alpha_{j,t} p_{j}}}{c_{i,t} R}$ for $i \neq j$ and $a_{ii}=1 - \frac{d_i}{R}$ are the entries of the row-stochastic matrix $A$. The auxiliary $\mathbf{z}_{i,t}$ are updated as $\mathbf{z}_{i,t+1}=\sum_{j=1}^{K} a_{ij} \mathbf{z}_{j,t}$ with initial $\mathbf{z}_{i,0}=\mathbf{e}_i$, where $\mathbf{e}_i$ is the $i$-th standard unit vector.


\subsection{Privacy Analysis}
The sensitivity of the received signal at node \( i \) due to neighbor \( j \) is computed by considering two neighboring datasets \( D_j \) and \( D_{j'} \) at node \( j \), differing in one data point. By combining \eqref{8} and \eqref{1}, the sensitivity of the signal transmitted from \( j \) to \( i \) at time \( t \) is given by
\begin{equation}\label{3}
\begin{aligned}
    \Delta_{ij}^{t}&= \max_{D_{j},D_{j'}} \|\mathbf{y}_{i}^{t}(D_{j})-\mathbf{y}_{i}^{t}(D_{j'})\|\\
    &= \max_{D_{j},D_{j'}} \||h_{ji}|\sqrt{\alpha_{j,t}p_{j}} \frac{\gamma_{t}}{z_{jj,t}} (\mathbf{g}_{j}^{t}(D_{j})-\mathbf{g}_{j}^{t}(D_{j'}))\|\\
    &\leq 2G\gamma_{t}\theta |h_{ji}|\sqrt{\alpha_{j,t}p_{j}},
\end{aligned}    
\end{equation}
with $ \frac{1}{z_{ii,t}} \leq \theta$, $\forall i \in \mathcal{V}$ and $\forall t$. Also, \( G \) is an upper bound of the gradient norms, i.e. $\|\mathbf{g}_{i,t} \| \leq G$, $\forall i \in \mathcal{V}$ and $\forall t$. By combining (\ref{4}) and (\ref{3}), we can state the following theorem to determine the privacy budget.

\begin{algorithm}[t]
\caption{Power-Controlled DP-DL}
\begin{algorithmic}[1] 
\State \textbf{Inputs:} Constraint set $\Omega$, number of epochs $T$, maximum powers of nodes, $p_{i}$, desirable maximum privacy threshold, $\epsilon_{\max}$, and learning rate $\{\gamma_t\}_{t=1}^T$ 
\State \textbf{Initialization:} Randomly initialize $\mathbf{x}_{i,0} \in \Omega$ and set $\mathbf{z}_{i,0} = \mathbf{e}_i$ for all $i \in \mathcal{V}$ and solve maximization problem \eqref{5} to find $\alpha_{i}$ and $\beta_{i}$
\For{each epoch $t = 0, 1, 2, \ldots, T-1$}
    \For{each node $i \in \mathcal{V}$ in parallel}
        \State Generate noise $\boldsymbol{\eta}_{i,t} \sim \mathcal{N}(0,\sigma_{i,t}^{2})$
        \State Construct $\tilde{\mathbf{x}}_{i,t}=\sqrt{\alpha_{i,t} p_{i}}\mathbf{x}_{i,t} + \sqrt{\beta_{i,t} p_{i}}\boldsymbol{\eta}_{i,t}$
        \State Multicast $\tilde{\mathbf{x}}_{i,t}$ and $\mathbf{z}_{i,t}$ to its neighbors $j \in \mathcal{N}_i$ 
        \Statex \hspace{1cm} through channels $h_{ij}$ and compute the aggregated 
        \Statex \hspace{1cm} signal $\mathbf{y}_{i,t}$ received from neighbors  
        \State Update $\mathbf{x}_{i,t+1}$ using (\ref{1})
        \State Update $\mathbf{z}_{i,t+1} = \sum_{j=1}^{K} a_{ij} \mathbf{z}_{j,t}$
    \EndFor
\EndFor
\State \textbf{Outputs} The sequence of model coefficients $\{\mathbf{x}_{i,t}\}_{i \in \mathcal{V}}$ for $t=1,2,...,T$
\end{algorithmic}
\end{algorithm}

\begin{theorem}\label{thm:eps_ij}
The proposed algorithm satisfies $(\epsilon_{ij} , \delta)$-DP for any node $i$ with respect to its neighbor $j$ in each epoch $t$, where $\epsilon_{ij}$ is given by
\begin{equation}\label{6}
\epsilon_{ij} = \frac{2G \gamma_t \theta |h_{ji}| \sqrt{\alpha_j p_j}}{\sqrt{\sum_{k \in N_i} |h_{ki}|^2 \beta_k p_k \sigma_{k,t}^2}} \sqrt{2 \ln \frac{1.25}{\delta}}.
\end{equation}
\end{theorem}

To achieve a trade-off between privacy and power efficiency, a heuristic optimization problem can be formulated to maximize the utility of model coefficients by determining the power fractions \(\alpha_i\) for multicasting signals from \(K\) clients, while satisfying a given upper bound on the privacy budget, \(\epsilon_{\max}\) for each epoch. The maximization problem is defined as
\begin{equation}\label{5}
\begin{aligned}
    &\max_{0 \leq \alpha_{j} \leq 1} \sum_{j=1}^K \alpha_{j}, \\
    \text{s.t. } &\epsilon_{ij} \leq \epsilon_{\max}, \quad \forall \, i,j \in \mathcal{V}.
\end{aligned}
\end{equation}

Here, \(\epsilon_{ij}\) denotes the privacy leakage associated with the signal received by client \(i\) from client \(j\), as defined in (\ref{6}). Using $\beta_j = 1 - \alpha_j$, the constraints in (\ref{5}) become linear, thus making it a linear program.  
Since $\sigma_{k,t}$ is assumed to be proportional to $\gamma_{t}$, (\ref{5}) is independent of $t$, allowing for solving it once.

The maximization problem ensures that the privacy budget for information received by any node in the system remains below the specified threshold, \(\epsilon_{\max}\), guaranteeing a given level of privacy in each epoch. Simultaneously, it maximizes the power fractions \(\alpha_j\), which are essential for efficient multicasting of signals.

\begin{corollary}
    Let $\epsilon_{\max}\stackrel{\triangle}{=}\max_{i,j}\epsilon_{ij}$ be the largest per-epoch privacy budget guaranteed by Theorem \ref{thm:eps_ij}, with common failure probability $\delta$.  By using the composition theorem, running the algorithm for $T$ epochs yields a mechanism that is $(T\epsilon_{\max},\,T\delta)$-DP.
\end{corollary}

The proposed algorithm is designed to run iteratively until convergence. The summary of this procedure is provided in Algorithm 1. 
\subsection{Convergence Rate Analysis}
As shown in \cite[Lem.~1]{MAI201994}, for all $i,j \in \mathcal{V}$  and $t$ there exists $0<\xi<1$ and $C>0$ such that $|[A^{t}]_{ij} - \pi_{j}| \leq C\xi^{t}$ and $|z_{ii,t}-\pi_i| \leq C\xi^{t}$. Thus, using Assumptions \ref{a:s_convex_fi} and \ref{a:omega}, we establish the convergence of our proposed algorithm. More specifically, by setting the learning rate $\gamma_t = \frac{1}{\mu \theta t}$, where $\mu > 0$ and $\theta > 0$, we ensure that the algorithm converges.

\begin{figure}[t]
  \centering
  \input{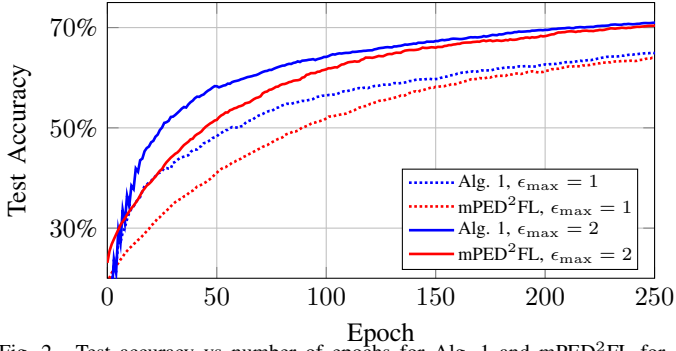}
  \vspace{-1cm}\caption{Test accuracy vs number of epochs for Alg.~1 and {mPED$^2$FL} for $\epsilon_{\max} = 1$ and $\epsilon_{\max} = 2$.}
  \label{fig:myfigure2}
\end{figure}

\begin{theorem}\label{thm:convergence}
The regret function of the proposed algorithm can be upper-bounded by
\begin{equation*}
\begin{aligned}
    \mathbb{E}[\mathbb{R}_{i}(T)] \leq U_{1} + U_{2}(1+\log T), 
\end{aligned}   
\end{equation*}
where $U_1$ and $U_2$ are defined as
\begin{equation*}
\begin{aligned}
    &U_{1} = \frac{\xi C G}{1-\xi} \Big (2(K+\theta) \sum_{i=1}^{K}\|x_{i,0}\| + K\theta L \Big ), \\
    &U_{2} = O\left(K^{3} m G^{2} \ln \frac{1.25}{\delta}\right)
\end{aligned}   
\end{equation*}
\end{theorem}

\begin{proof}
    Using Lemmas 3 and 4 from \cite{9013030} and setting $\mathbf{x}=\mathbf{x}^{*}$, we can find the upper bound of $\sum_{i=1}^{K}\mathbb{E}[f_{i}(\mathbf{x}_{i,t}) - f_{i}(\mathbf{x}^{*})]$. Since our work introduces a DP mechanism that injects scaled Gaussian noise, $\sqrt{\frac{\beta_{i,t}}{\alpha_{i,t}}}\boldsymbol{\eta}_{i,t}$ as opposed to \cite{9013030} which uses Laplace noise, computing upper bounds for $\mathbb{E}[\|\sqrt{\frac{\beta_{i,t}}{\alpha_{i,t}}}\boldsymbol{\eta}_{i,t}\|]$ and $\mathbb{E}[\|\sqrt{\frac{\beta_{i,t}}{\alpha_{i,t}}}\boldsymbol{\eta}_{i,t}\|^{2}]$ will complete the proof of Theorem~2. For brevity, we present $U_2$ in big-O notation only.
\end{proof} 

\section{Simulation Results}
We evaluate Alg.~1 on the MNIST dataset, distributed across $K=4$ clients in a fully connected topology. The entire dataset is first divided into training (80\%) and test (20\%) datasets. The training set is further partitioned in a non-iid manner among the clients. The training set is sorted by class labels and divided equally, with each 25\% allocated to clients 1 through 4 sequentially. The test set is shared and accessible to all clients. After completing the training phase, all clients evaluate their respective models on the test set to measure performance.

To evaluate the performance of our proposed algorithm, we compare it with the PED$^2$FL algorithm \cite{10025677}. However, since PED$^2$FL assumes a doubly stochastic matrix and considers identical channel gains for all neighbors, i.e., $h_{ij} = h_i$, we modified the algorithm to ensure a fair comparison. This modification involves creating separate compensation factors $c_i$ (one for each client) instead of using a single $c$ in \cite[eq.~(4)]{10025677}. This modification is necessary to account for heterogeneous channel gains $h_{ij} \neq h_{ik}$. As a result of this modification, the algorithm, which we refer to as mPED$^2$FL, now uses different power scaling factors $\alpha_{ij}$ instead of a single $\alpha_i$. Therefore, mPED$^2$FL cannot multicast the same signal to all neighbors. In contrast, our proposed algorithm employs the same $\alpha_i$ for all neighbors of a node $i$, enabling efficient multicasting and hence reducing the number of channel uses.

\begin{figure}[t]
  \centering
  \input{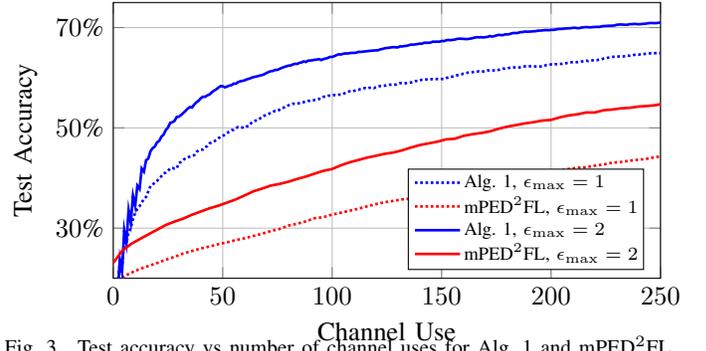}
  \vspace{-1cm}\caption{Test accuracy vs number of channel uses for Alg.~1 and {mPED$^2$FL} for $\epsilon_{\max} = 1$ and $\epsilon_{\max} = 2$.}
  \label{fig:myfigure3}
\end{figure}

For the numerical simulations, the channel gains are randomly selected from a uniform distribution and fixed:
\[\left(h_{ij}^{(1)}\right) = 
\setlength{\arraycolsep}{3pt} 
\renewcommand{\arraystretch}{0.8} 
\begin{pmatrix}
0 & 0.92 & 0.94 & 0.98 \\
0.92 & 0 & 0.92 & 0.96 \\
0.92 & 0.96 & 0 & 0.95 \\
0.88 & 0.92 & 0.98 & 0
\end{pmatrix}
\]

\[\left(h_{ij}^{(2)}\right) =
\setlength{\arraycolsep}{3pt} 
\renewcommand{\arraystretch}{0.8} 
\begin{pmatrix}
0 & 0.92 & 0.94 & 0.98 \\
0.92 & 0 & 0.92 & 0.96 \\
0.95 & 0.943 & 0 & 0.95 \\
0.95 & 0.96 & 0.98 & 0
\end{pmatrix}
\]

The channel gains $h_{ij}^{(1)}$ and $h_{ij}^{(2)}$ are used for two cases of privacy budget: $\epsilon_{\max}=1$ and $\epsilon_{\max}=2$, correspondingly. In both cases, all nodes are assumed to have an equal maximum power $p_i=1$. Additionally, the learning rate is given as $\gamma_{t} = \frac{1}{\sqrt{t}}$, and the noise standard deviation is set as $\sigma_{i,t} = \frac{10}{\sqrt{t}}$. These settings are used consistently in both our proposed algorithm and the mPED$^2$FL algorithm to ensure a fair comparison. Also, each client performs a multinomial logistic regression model locally on its assigned training data in both algorithms. To compare our proposed algorithm with mPED$^2$FL we first compute the maximum of all $\epsilon_{ij}$ values ($\epsilon_{\max}$) in mPED$^2$FL. Using this value, we then solve the maximization problem (\ref{5}) to determine the maximized $\alpha_i$ for our proposed algorithm. 

Fig. \ref{fig:myfigure2}, compares our proposed algorithm Alg.~1 with mPED$^2$FL for both $\epsilon_{\max} = 1$ and $\epsilon_{\max} = 2$. The results demonstrate that our algorithm performs better than mPED$^2$FL in terms of both accuracy and rate of convergence. Additionally, Fig. \ref{fig:myfigure3} compares the two algorithms based on the number of channel uses. Since mPED$^2$FL cannot multicast, each client must individually send signals to all its neighbors, resulting in four channel uses per epoch, as opposed to single channel use in our approach. Consequently, our proposed algorithm converges faster.



\clearpage

\balance
\bibliographystyle{IEEEtran}
\bibliography{Reference}









\end{document}